\documentclass[11pt]{article}
\usepackage[margin=1in]{geometry}

\usepackage{multirow}
\usepackage{authblk}
\usepackage[utf8]{inputenc} 
\usepackage[T1]{fontenc}    
\usepackage{hyperref}       
\usepackage{url}            
\usepackage{booktabs}       
\usepackage{amsfonts}       
\usepackage{nicefrac}       
\usepackage{microtype}      
\usepackage{xcolor}         
\usepackage{times}
\usepackage{amssymb, amsmath, amsthm}
\usepackage{mathtools}
\usepackage{color}
\usepackage{graphicx}
\usepackage{verbatim} 
\usepackage{url} 
\usepackage[ruled,vlined, linesnumbered]{algorithm2e}
\usepackage{balance}
\allowdisplaybreaks

\SetKwComment{Comment}{/* }{*/}
\SetKwInput{Init}{Initialize}
\SetKwProg{When}{When}{:}{}
\SetKwBlock{StartRound}{In the beginning of one round:}{end}
\SetKwBlock{StartPhase}{In the beginning of one phase:}{end}
\SetKwBlock{EndF2Round}{When local F2 increases by a factor of $2$:}{end}
\SetKwBlock{Parallel}{Do in parallel:}{end}
\SetKwBlock{FirstRound}{First Round:}{end}
\SetKwBlock{SecondRound}{Second Round:}{end}

\SetKw{Continue}{continue}

\newtheorem{definition}{Definition}
\newtheorem{theorem}{Theorem}
\newtheorem{lemma}{Lemma}
\newtheorem*{remark}{Remark}

\newcommand{\eps}{\varepsilon}
\newcommand{\polylog}{\mathrm{polylog}}
\newcommand{\E}{\mathsf{E}}

\newcommand{\var}{\mathsf{Var}}
\renewcommand{\Pr}{\mathsf{Pr}}

\newcommand{\alg}{\mathsf{ALG}}

\newcommand{\lar}{\leftarrow}

\newcommand{\wtilde}[1]{\widetilde{#1}}
\newcommand{\what}[1]{\widehat{#1}}

\renewcommand{\paragraph}[1]{\medskip\noindent{\bf #1} }

\newcommand{\xomit}[1]{}

\newcommand{\ifconf}[1]{}

\newcommand{\remove}[1]{}

\newcommand{\poly}{{\mathrm{poly}}}

\newcommand{\Rar}{\Rightarrow}

\def\FullBox{\hbox{\vrule width 8pt height 8pt depth 0pt}}
\def\qed{\ifmmode\qquad\FullBox\else{\unskip\nobreak\hfil
\penalty50\hskip1em\null\nobreak\hfil\FullBox
\parfillskip=0pt\finalhyphendemerits=0\endgraf}\fi}

\begin{document}

\title{\Large{\bf{Simple and Optimal Algorithms for Heavy Hitters and Frequency Moments in Distributed Models}}}

\author[a,b]{Zengfeng Huang} 
 \author[a]{Zhongzheng Xiong}
 \author[a]{Xiaoyi Zhu}
 \author[c]{Zhewei Wei}
 \affil[a]{School of Data Science, Fudan University}
 \affil[b]{Shanghai Innovation Institute}
 \affil[c]{Renmin University of China}

\maketitle

\begin{abstract}
We consider the problems of distributed heavy hitters and frequency moments in both the coordinator model and the distributed tracking model (also known as the distributed functional monitoring model). We present simple and optimal (up to logarithmic factors) algorithms for $\ell_p$ heavy hitters and $F_p$ estimation ($p \geq 2$) in these distributed models.

For $\ell_p$ heavy hitters in the coordinator model, our algorithm requires only one round and uses $\tilde{O}(k^{p-1}/\eps^p)$ bits of communication. For $p > 2$, this is the first near-optimal result. By combining our algorithm with the standard recursive sketching technique, we obtain a near-optimal two-round algorithm for $F_p$ in the coordinator model, matching a significant result from recent work by Esfandiari et al.\ (STOC 2024). Our algorithm and analysis are much simpler and have better costs with respect to logarithmic factors. Furthermore, our technique provides a one-round algorithm for $F_p$, which is a significant improvement over a result of Woodruff and Zhang (STOC 2012).

Thanks to the simplicity of our heavy hitter algorithms, we manage to adapt them to the distributed tracking model with only a $\polylog(n)$ increase in communication. For $\ell_p$ heavy hitters, our algorithm has a communication cost of $\tilde{O}(k^{p-1}/\eps^p)$, representing the first near-optimal algorithm for all $p \geq 2$. By applying the recursive sketching technique, we also provide the first near-optimal algorithm for $F_p$ in the distributed tracking model, with a communication cost of $\tilde{O}(k^{p-1}/\eps^2)$ for all $p \geq 2$. Even for $F_2$, our result improves upon the bounds established by Cormode, Muthukrishnan, and Yi (SODA 2008) and Woodruff and Zhang (STOC 2012), nearly matching the existing lower bound for the first time.
\end{abstract}

\setcounter{page}{-1}
\thispagestyle{empty}
\newpage
\tableofcontents
\thispagestyle{empty}
\newpage

\section{Introduction}

In many modern big data applications, data is collected and stored across a large number of nodes~\cite{chandramouli2013supporting, madden2005tinydb, schiller2004location}. In such distributed systems, communication costs are often a primary concern. Designing distributed algorithms with optimal communication costs is therefore of significant interest, both theoretically and practically. A classic model of distributed computation is the number-in-hand \emph{coordinator model}~\cite{dolev1989multiparty}, in which a coordinator wish to evaluate a function $f(v^{(1)},\cdots, v^{(k)})$. The input $v^{(1)},\cdots,v^{(k)}$ is distributed across $k$ sites, with $v^{(i)}$ only known to site $i$.  
Extensive research has focused on understanding the communication and round complexity of fundamental problems in this model (e.g.,\cite{woodruff2013distributed, braverman2017rounds, woodruff2014optimal, assadi2022rounds, huang2021communication, huang2017communication, viola2015communication, phillips2012lower}). Another well-studied model is the distributed tracking model, also known as distributed functional monitoring. In this model, each site $i$ receives items over time in a streaming manner. Let $S_i(t)$ denote the stream observed by site $i$ up to time $t$. The coordinator aims to track the value of a function $f$ defined over the multi-set union of ${S_i(t)|~ i = 1, \ldots, k}$ at all times. This model naturally combines aspects of both the coordinator and streaming models, and is highly motivated by distributed system applications. Numerous functions have been studied in this model (e.g.~\cite{dilman2001efficient, cormode2005holistic, keralapura2006communication, cormode2008algorithms, yi2009optimal, arackaparambil2009functional, tirthapura2011optimal, cormode2012continuous, woodruff2012tight, chen2017improved, wu2020learning}).

In this paper, we consider two fundamental and related problems: heavy hitters (HH) and frequency moments, in both the coordinator model and the distributed tracking model.
In the coordinator model (static setting), each site $i \in [k]$ holds a non-negative frequency vector $v^{(i)} \in \mathbb{R}^n$. For the $F_p$ problem, the coordinator needs to estimate the $F_p$ value of the overall frequency vector, defined as $F_p = \|\sum_{i=1}^k v^{(i)}\|_p^p$, with a relative error of $\eps$. For the $\ell_p$-HH problem, the coordinator must estimate the frequency of each element with an additive error of $\eps \ell_p$.
In the tracking model (dynamic setting), each site receives a stream of elements, causing its local frequency vector to change over time. The coordinator’s task is to continuously track the $F_p$ value and the frequency of each element. In this paper, we focus exclusively on cases where $p \geq 2$.

The $F_p$ problem in distributed models was first introduced by Cormode et al.~\cite{cormode2008algorithms}. In the static setting, $F_2$ can be easily solved using the AMS sketch~\cite{alon1996space} with a communication cost of $\tilde{O}\left(k/\eps^2\right)$,\footnote{We use $\tilde{O}$ to suppress $\polylog(n)$ factors, and we treat $p$ as a constant, as in previous work.} which is known to be optimal~\cite{woodruff2012tight}. For $p > 2$, a straightforward application of sketching algorithms leads to a communication cost with an $n^{1-\frac{2}{p}}$ dependence. Woodruff and Zhang~\cite{woodruff2012tight} were the first to propose an algorithm that eliminated this $n^{1-\frac{2}{p}}$ dependence, revealing a unique characteristic of distributed $F_p$ compared to its streaming counterpart. Their algorithm achieves $\tilde{O}(k^{p-1}/\poly(\eps))$ communication. They also proved a lower bound of $\Omega(k^{p-1}/\eps^2)$, leaving the problem of closing the gap between upper and lower bounds as a significant open question. Recently, Esfandiari et al.~\cite{esfandiari2024optimal} made significant progress in the coordinator model by proposing a two-round protocol with a near optimal communication cost of $\tilde{O}(\frac{k^{p-1}}{\eps^2})$. They also proved that any one-round protocol must incur at least $\Omega(\frac{k^{p - 1}}{\eps^p})$ bits of communication. Thus, determining the optimal one-round communication complexity for $p > 2$ remains an interesting open problem.

Tracking functions in the distributed tracking model can be significantly more challenging than in the static coordinator model, as illustrated by the $F_2$ problem. In the distributed tracking model, Cormode et al.~\cite{cormode2008algorithms} focused on the case where $p = 2$ and proposed an algorithm with a communication cost of $\tilde{O}\left(\frac{k^2}{\eps} + \frac{k^{1.5}}{\eps^3} \right)$. Woodruff and Zhang's algorithm~\cite{woodruff2012tight} further improved the dependence on $k$ for $F_2$ with a cost of $\tilde{O}(k^{p-1}/\varepsilon^{\Theta(p)})$ bits, but at the expense of a much higher $1/\varepsilon^{\Theta(p)}$ factor. As a result, determining the optimal communication bound for estimating $F_p$ in the distributed tracking model remains an open problem, even for $p = 2$. This is somewhat surprising, given that $F_2$ is one of the most extensively studied problems in related models.

For heavy hitter problems in the coordinator model, the case of $p = 2$ can be easily solved using the count sketch algorithm~\cite{charikar2002finding} with a communication cost of $\tilde{O}(k/\eps^2)$. However, to the best of our knowledge, there is no existing theoretical work that explicitly studies general $\ell_p$-HH in distributed models. In the streaming model, $\ell_p$-HH can be reduced to $\ell_2$-HH using Hölder's inequality (e.g.,~\cite{braverman2013generalizing}), with the space complexity proportional to $n^{1-\frac{2}{p}}$, which is optimal. Applying this reduction in the coordinator model leads to a protocol with a communication cost that scales with $n^{1-\frac{2}{p}}$. Similar to the $F_p$ problem, we try to answer the question whether this $n^{1-\frac{2}{p}}$ factor can be eliminated. In the distributed tracking model, the optimal complexity of  $\ell_2$-HH is still not settled.

\subsection{Our Results}
\paragraph{Coordinator model.} We present a simple sampling algorithm for $\ell_2$-HH that improves upon the count sketch approach. Specifically, with a communication cost of $\tilde{O}(k/\eps^2)$ (matching that of count sketch), our algorithm estimates the frequency of each element with an additive error of $\eps \ell_2'$, where $\ell_2' = \left(\sum_{i=1}^k \|v^{(i)}\|_2^2\right)^{1/2}$. In comparison, the accuracy guarantee for count sketch is $\eps \ell_2$, where $\ell_2 = \|\sum_{i=1}^k v^{(i)}\|_2$. Notably, $\ell_2' \leq \ell_2$ and can be smaller by a factor of up to $\sqrt{k}$. Although these values are the same in the worst case, this improvement is critical to our entire framework. This result enables us to design a new one-round algorithm for $\ell_p$-HH using a straightforward thresholding technique combined with the standard reduction via Hölder's inequality. The communication cost of our new $\ell_p$-HH algorithm is $\tilde{O}(k^{p-1}/\eps^p)$. For $p > 2$, this is the first near-optimal result.

In the streaming literature, a standard technique known as recursive sketching~\cite{indyk2005optimal,braverman2013generalizing} reduces $F_p$ estimation to $\ell_p$-HH. By applying this reduction as a black box, we obtain a new two-round algorithm for distributed $F_p$ in the coordinator model. Our communication bound matches that of \cite{esfandiari2024optimal}, specifically $\tilde{O}(k^{p-1}/\eps^2)$, but with a significantly smaller  $\polylog(n)$ factor. It is worth noting that our algorithm and its analysis are considerably simpler than those presented in \cite{esfandiari2024optimal}. Furthermore, it is relatively straightforward to derive a one-round algorithm from our framework. The communication cost for this algorithm is $\tilde{O}(k^{p-1}/\eps^{p+2})$, which represents a significant improvement over the best previous result from \cite{woodruff2012tight} and matches the lower bound established by \cite{esfandiari2024optimal} up to a $1/\eps^2$ factor. Closing this gap remains an intriguing open question. We remark that this reduction from $F_p$ to $\ell_p$-HH together with the lower bound of~\cite{woodruff2012tight} for $F_p$ also proves a lower bound of $\Omega( k^{p-1}/\eps^p)$ for $\ell_p$-HH, which shows our upper bound is near optimal.

\paragraph{Distributed tracking model.}
According to our pipeline in the static setting, the key is to have a heavy hitter algorithm that guarantees an additive error of $\eps \ell_2'$. We adapt our static sampling algorithm to get a tracking algorithm for $\ell_2$-HH that maintains the same error guarantee and communication cost (with one additional logarithmic factor). While our static algorithm is extremely simple, the probability of sampling each element is a quadratic function of its corresponding frequency. This  makes the dynamization non-trivial, prompting us to develop a general technique for sampling from non-linear probabilities in the dynamic tracking model. Having addressed the $\ell_2$-HH case, we can easily extend our approach to solve the general $\ell_p$-HH tracking problem using the same reduction, resulting in a communication cost of $\tilde{O}(k^{p-1}/\eps^p)$.

For $F_p$ tracking, we still need to dynamize the recursive sketching technique, which is non-trivial if we aim to achieve optimal bounds. To address this, we first refine the notion of $(\alpha,\eps)$-cover in recursive sketching as introduced by Braverman et al.~\cite{braverman2013generalizing}, and then we develop a cover tracking algorithm by carefully combining several techniques from the distributed tracking literature. The communication cost of our final algorithm is $\tilde{O}(k^{p-1}/\eps^2)$, which significantly improves upon previous results, even for $p=2$, and matches the lower bound established by Woodruff and Zhang~\cite{woodruff2012tight} up to logarithmic factors.

Overall, we present a simple and general algorithmic framework for heavy hitters and frequency moments in distributed models. From this framework, we get near-optimal bounds for all the fundamental problems studied in this paper, with the only exception of one-round $F_p$.

\begin{table*}[h]
\centering
\caption{A summary of the our results. 1-round and 2-round algorithms are both within the static setting.}
\begin{tabular}{|l|ccc|}
\hline
\textbf{Problem}                          & \multicolumn{1}{c|}{\textbf{1-round algorithm}} & \multicolumn{1}{c|}{\textbf{2-round algorithm}} & \multicolumn{1}{c|}{\textbf{Tracking}} \\ \hline
$\ell_p$ heavy hitter  & \multicolumn{3}{c|}{ $\Tilde{\Theta}\left(\frac{k^{p-1}}{\eps^p}\right)$(Theorem \ref{thm:lp_hh_static}, Theorem \ref{thm:lp_hh_tracking})}  \\ \hline
\multirow{2}{*}{$F_p$ estimation (previous)} & \multicolumn{1}{c|}{UB: $\Tilde{O}\left(\frac{k^{p-1}}{\eps^{\Theta(p)}}\right)$\cite{woodruff2012tight}}   &  \multicolumn{1}{c|}{UB: $\Tilde{O}\left(\frac{k^{p-1}}{\eps^2}\right)$\cite{esfandiari2024optimal}}   & \multicolumn{1}{c|}{UB: $\Tilde{O}\left(\frac{k^{p-1}}{\eps^{\Theta(p)}}\right)$\cite{woodruff2012tight}}    \\ & \multicolumn{1}{c|}{LB: $\Tilde{\Omega}\left(\frac{k^{p-1}}{\eps^p}\right)$\cite{esfandiari2024optimal}} & \multicolumn{1}{c|}{LB: ${\Omega}\left(\frac{k^{p-1}}{\eps^2}\right)$\cite{woodruff2012tight}}   &  \multicolumn{1}{c|}{LB: ${\Omega}\left(\frac{k^{p-1}}{\eps^2}\right)$\cite{woodruff2012tight}} \\ \hline
\multirow{2}{*}{$F_p$ estimation (this paper)}                & \multicolumn{1}{c|}{UB: $\Tilde{O}\left(\frac{k^{p-1}}{\eps^{p + 2}}\right)$}                           & \multicolumn{1}{c|}{UB: $\Tilde{O}\left(\frac{k^{p-1}}{\eps^2}\right)$}                           &     \multicolumn{1}{c|}{UB: $\Tilde{O}\left(\frac{k^{p-1}}{\eps^2}\right)$} \\ 
 & \multicolumn{1}{c|}{(Theorem \ref{thm:fp_static})} & \multicolumn{1}{c|}{(Theorem \ref{thm:fp_static})}   &  \multicolumn{1}{c|}{(Theorem \ref{thm:fp_tracking})} \\ \hline
\end{tabular}
\label{tab:summary}
\end{table*}

\subsection{Technical Overview}
\paragraph{Distributed $\ell_p$ heavy hitter estimation} Recall that the goal of $\ell_p$ heavy hitter estimation is to estimate each $v_j$ within $\eps \ell_p$ additive error, i.e., $|\hat{v}_j - v_j|\leq \eps \ell_p$, where $\ell_p = (\sum_{j =1}^n(\sum_{i =1}^k v_{ij})^p)^{1/p}$. Define $\ell_p' = (\sum_{i=1}^k\sum_{j = 1}^n v_{ij}^p)^{1/p}$. Note that $\ell_p' \le \ell_p$ by the non-negativity of each $v_{ij}$ and $\ell_p'$ can be computed exactly with only $O(k\log n)$ bits (each site sends its local $F_p$ value to the coordinator). 

We will see that at the heart of results is a simple $\ell_2$-HH algorithm, which estimates each frequency with additive error $\eps \ell_2'$ rather than $\eps \ell_2$ as for count sketch. Our algorithm is a sampling-based method. For site $i$ and element $j$, $v_{ij}$ is sent to the coordinator with probability $p_{ij} = \frac{v_{ij}^2}{\eps^2 F_2^{(i)}}$ where $F_2^{(i)}= \sum_{j=1}^n v_{ij}^2$ is the local $F_2$ value of site $i$. Note that for each $v_{ij}$,  $\hat{v}_{ij} = \mathbf{1}(\text{site $i$ sends $v_{ij}$})\cdot v_{ij}/p_{ij}$ is an unbiased estimator of $v_{ij}$ with variance less than $\eps^2 F_2^{(i)}$. Thus, $\sum_{i =1}^k \hat{v}_{ij}$ is an unbiased estimator of $v_{j}$ with variance less than $\eps^2 \sum_{i =1}^k F_2^{(i)} = \eps^2 (\ell_2')^2$. By Chebyshev's inequality, with constant probability $|\hat{v}_j - v_j| \leq \eps \ell_2'$. Note that the expected communication cost of one site is $\sum_{j = 1}^n \frac{v_{ij}^2}{\eps^2 F_2^{(i)}} =  \frac{1}{\eps^2}$ and thus the total communication cost is $\frac{k}{\eps^2}$. 

For general $\ell_p$-HH estimation, the classical reduction to $\ell_2$-HH works as follows. By setting $\eps' = \frac{\eps \ell_p'}{\ell_2'}$ in the $\ell_2$-HH algorithm, we obtain the bound $|\hat{v}_j - v_j| \leq \eps \ell_p'$, with a communication cost of $\tilde{O}\left(\frac{k}{\eps^2} \cdot \left(\frac{\ell_2'}{\ell_p'}\right)^2\right)$. According to Hölder's inequality, we have $\frac{\ell_2'}{\ell_p'} \leq (nk)^{1-\frac{2}{p}}$, resulting in a polynomial dependence on $n$. Hölder's inequality is generally tight, and this reduction yields tight results in the streaming model. However, it is unclear whether this reduction is inherently inefficient in distributed settings.

In this paper, we show that our improved guarantee on $\ell_2$-HH effectively mitigates this disadvantage through a simple thresholding trick: Each site ignores all $v_{ij}$ such that $v_{ij} < \frac{\eps \ell_p'}{k}$, which incurs an error bounded by $\eps \ell_p'$. According to a simple technical lemma (Lemma~\ref{lem:l2-lp-inquality}), if we set all such $v_{ij}$'s to $0$, we achieve a new bound of $\frac{\ell_2'}{\ell_p'} \leq \frac{k^{p-2}}{\eps^{p-2}}$, resulting in a communication cost of $O\left(\frac{k^{p-1}}{\eps^p}\right)$, as desired. We note that if we apply an $\ell_2$-HH algorithm that only provides an error of $\eps \ell_2$, the same as the count sketch, the above argument would yield a suboptimal dependence on $k$, even if we only require an $\eps \ell_p$ error for $\ell_p$-HH. Furthermore, it is the simplicity of our $\ell_2$-HH algorithm that enables its adaptation  in the distributed tracking model, which is a crucial technical building block of our distributed tracking results.

The tracking algorithm for $\ell_2$-HH is more technical. In the static setting, each $v_{ij}$ is sent to the coordinator with probability $p_{ij}\approx \frac{v_{ij}^2}{\eps^2 F_2^{(i)}}$, which is quadratic in $v_{ij}$. This non-linearity is the source of difficulty in dynamizing this sampling algorithm.  In this paper, we propose a general technique for dynamizing static distributed sampling algorithms with non-linear sampling probability. At a high level, we substitute the non-linear sampling function with a piece-wise linear approximation and divide the whole monitoring period into multiple intervals, each with a linear sampling function. The actual algorithm that implements this idea is quite technical. A more detailed overview is provided in Section~\ref{sec:l2-hh-tracking}.

\paragraph{Distributed $F_p$ estimation}. In the coordinator model, with our $\ell_p$-HH algorithms, it is fairly easy to directly apply the recursive sketching technique to get optimal $F_p$ algorithms. The key step in recursive sketching is to compute a heavy hitter set, which can be done using a HH algorithm. But for each element in the set, its frequency estimation needs to have a higher accuracy than the HH algorithm. For a two-round protocol,  we can use a second round of communication to get exact frequencies for all elements selected in the first round. For a one-round algorithm, we need to set the error parameter in the HH algorithm much smaller, hence incurring much higher communication. Even so, the communication cost of our one-round algorithm is significantly better than previous results.

Designing algorithms for optimally tracking  $F_p$ is technically more involved. Tracking a heavy hitter set is relatively easy. The main difficulty is to track the value of each heavy hitter with relative $\eps$-error. The size of the heavy hitter set is $O(1/\eps^2)$. If we track the frequency of each heavy hitter separately using the optimal algorithm from~\cite{huang2012randomized}, the total cost is $\tilde{O}\left(\frac{1}{\eps^2}\cdot (k+\frac{\sqrt{k}}{\eps})\right)$, which is not optimal unless $\eps$ is large (Recall, for $F_2$ our goal is $\tilde{O}(k/\eps^2)$). To get an optimal bound, we revisit the proof of recursive sketching in \cite{braverman2013generalizing}, and notice that the correctness holds with a much relaxed condition on frequency estimation. Based on this observation, we design a more communication-efficient algorithm by combining techniques from \cite{cormode2008algorithms, xiong2024adversarially}. 

\section{Preliminaries}

\paragraph{Notation.} We always use $k$ to denote the number of sites.  The universe of data items is $[n] := {1, 2, \ldots, n}$, and let $m$ be the length of the entire data stream, which is the union of the $k$ distributed streams. As in previous research~\cite{esfandiari2024optimal}, we assume that $m = \mathrm{poly}(n)$, or $\log m =O(\log n)$.
We represent the frequency vector of the overall stream as $v = [v_1, v_2, \ldots, v_n] \in \mathbb{Z}^n$, where $v_j$ is the total frequency of element $j$ across all sites. For each site $i$, we use $v^{(i)}$ to denote the local frequency vector, with $v_{ij}$ representing the frequency of element $j$ at site $i$. Clearly, $v= \sum_{i=1}^k v^{(i)}$ and  $v_j = \sum_{i=1}^k v_{ij}$. The $p$-th moment of the frequency vector $v$ is defined as $F_p(v) = \sum_{j=1}^n v_j^p$, and the $\ell_p$ norm of $v$ is $\ell_p(v) = \left(\sum_{j=1}^n v_j^p\right)^{1/p}$. When the context is clear, we may simply write $\ell_p$ and $F_p$ to denote the $\ell_p$ norm and $p$-th moment of $v$, respectively.  In the distributed setting, given local frequency vectors $v^{(1)},\cdots, v^{(k)}$, we can easily get an underestimate of $F_p$. We define $F_p'(v^{(1)},\cdots, v^{(k)}) = \sum_{i=1}^k F_p(v^{(i)})= \sum_{i=1}^k \sum_{j=1}^n v_{ij}^p  $ .  When the decomposition $v= \sum_{i=1}^k v^{(i)}$ is clear from the context, we simply use $F'_p(v)$ to denote the above quantity.  
 Similarly, define $\ell_p'(v) = F_p'(v)^{1/p}$.  It is easy to see $F'_p\le F_p$, since all local frequency vectors are non-negative. In the tracking setting, we use $v(t)$ to represent the frequency vector at time $t$, and $w(t_1:t_2)$ to denote the increment vector, which captures the difference between the frequency vectors at times $t_2$ and $t_1$, i.e., $w(t_1:t_2) = v(t_2) - v(t_1)$.

\paragraph{Problem definition.} The goal of the distributed heavy hitter estimation problem is to provide an estimate $\hat{v}_j$ for each $v_j$ such that $|\hat{v}_j - v_j| \leq \eps \ell_p(v)$ holds with constant probability. The distributed $F_p$ estimation problem aims to estimate $F_p(v)$ such that with constant probability, its estimate $\hat{F_p}(v)$ satisfies $|\hat{F_p}(v) - F_p(v)| \leq \eps F_p(v)$. In the distributed tracking model, the above estimation goals should be met at any time step. Specifically, for the distributed heavy hitter problem, for any time $t$ and any $v_j(t)$, with constant probability, $|\hat{v}_j(t) - v_j(t)| \leq \eps \ell_p(v(t))|$. For $F_p$ estimation, for any time $t$, $|\hat{F_p}(v(t)) - F_p(v(t))| \leq \eps F_p(v(t))$ holds with constant probability. The complexity measure of interest is the communication complexity, defined as the total number of bits communicated between the coordinator and all sites throughout the estimation process.

\paragraph{Recursive sketching}. Recursive sketching from Braverman et al.~\cite{braverman2013generalizing} is a general technique for approximating functions such as frequency moments ($p \geq 2$) in the streaming model. It generalizes and improves the layering technique in \cite{indyk2005optimal}. Recursive sketching reduces the problem of estimating a vector function to heavy hitter problems.
\begin{definition}[($\alpha, \eps$)-cover]
\label{def:hh_cover}
A non-empty set of (index, value) pairs $Q = \{(i_1, w_1),\ldots, (i_t, w_t)\}$  is called an $(\alpha, \eps)$-cover w.r.t.\ the non-negative vector $u$, if the following two conditions hold:
\begin{enumerate}
    \item  $Q$ contains all heavy hitters: $\forall i \in [n]$, if $u_i > \alpha \sum_{i}u_i$, then $i \in Q$.
    \item  Value approximation: $\forall j \in [t]$, $(1 - \eps) u_{i_j} \leq w_j \leq (1+\eps)u_{i_j}$.
\end{enumerate}
$\operatorname{Ind}(Q)$ denotes the index set $\{i_1, i_2, \ldots, i_t\}$ of $Q$. 
\end{definition}

Let $u\in \mathbb{N}^n$ be a vector, and $h\in \{0, 1\}^n$ be a random binary vector. We define $u_h$ as the sub-vector of $u$ only containing the elements $e$ such that $h_e = 1$. Given two binary vectors $h_1,h_2$, define $h_1 \cdot h_2$ as the entry-wise product of $h_1$ and $h_2$. Typically $u$ is a vector derived from the frequency vector $v$, e.g., for the $F_p$ problem, $u_j = v_j^p$, where $v_j$ is the frequency of element $j$. Let $|u|$ denote the $\ell_1$ norm of $u$. 
In this paper, we only focus on computing $(\alpha, \eps)$-covers where the underlying frequency vector $v$ is a sum of $k$ vectors distributed across $k$ sites.
We use $\mathrm{DC}(u, \alpha, \eps, \delta)$ to denote a distributed algorithm that produces an $(\alpha, \eps)$-cover w.r.t.\ $u$ w.p.\ $1 - \delta$ and let $\mu(u, \alpha,\eps,\delta)$ be the corresponding communication cost. The recursive sketching algorithm for estimating $|u|$ is presented in Algorithm~\ref{alg:recursive_sktech}. 
\begin{algorithm}
\caption{Recursive Sketching~\cite{braverman2013generalizing}}
\label{alg:recursive_sktech}
Generate \(\phi=O(\log (n))\) independent binary random vectors \(h_1, \ldots, h_\phi\) using public randomness.\\
Let $u^l= u_{h_1\cdot h_2\cdots h_l}$ \Comment*[r]{only contains $e$ such that $h_{1,e}=\cdots h_{l,e} = 1$ }
Compute, in parallel, \(Q_l=\mathrm{DC}\left(u^l, \frac{\epsilon^2}{\phi^3}, \epsilon, \frac{1}{\phi}\right)\). \\
If \(F_0\left(u^\phi\right)>100\) then output $0$ and stop. Otherwise compute \(Y_\phi=\left|u^\phi\right|\) exactly.\\
\For{\(l=\phi-1, \ldots, 0\)}{
\(Y_l=2 Y_{l+1}-\sum_{(i,w)\in Q_l}\left(1-2 h_{l+1,i}\right) w\).}
Output \(Y_0\).\\
\end{algorithm}
\begin{theorem}
\label{thm:recursive_sketch}
Recursive sketching (Algorithm \ref{alg:recursive_sktech}) outputs a \((1 \pm \eps)\)-approximation of \(|u|\)  w.p.\ at least 0.9. The Communication cost is \(O\left(\log (n) \cdot \mu\left(n, \frac{\eps^2 }{\log ^3(n)}, \eps, \frac{1}{\log (n)}\right)\right)\) bits.
\end{theorem}
\begin{remark}
    We note that the proof of the above theorem in \cite{braverman2013generalizing} only requires $|\sum_{j\in Q} u_{i_j} - w_j| \leq \epsilon |u|$. Therefore, the second condition of the $(\alpha, \epsilon)$-cover can be relaxed. We will use this observation in section \ref{sec:fp_tracking}.
\end{remark}




\section{Distributed Heavy Hitters Estimation}\label{sec:HH}

\subsection{Heavy Hitters in Coordinator Model}
\subsubsection{$\ell_2$-HH in Coordinator Model}
Our $\ell_2$ heavy hitter algorithm is a sampling-based approach. The goal is to find an unbiased estimate for each $v_j$ with a variance less than $\eps^2F_2(v)$. Note that $v_j = \sum_{i = 1}^k v_{ij}$. If we have an unbiased estimation $\hat{v}_{ij}$ for each $v_{ij}$ with a variance bounded by $O(\eps^2 F_2(v^{(i)}))$, then $\sum{\hat{v}_{ij}}$ will meet our objective. To this end, site $i \in [k]$ sends $v_{ij}$ with probability $p_{ij} = O(\frac{v_{ij}^2}{\eps^2 F_2(v^{(i)})})$. Through direct calculation, we can see that $\hat{v}_{ij} = \mathbf{1}(\text{Site $i$ sends $v_{ij}$})\cdot v_{ij}/p_{ij}$ satisfies our requirement. The expected communication cost per site is $\log n \cdot \sum_{j=1}^n \frac{v_{ij}^2}{\eps^2F_2(v^{(i)})} = O(\frac{\log n }{\eps^2})$, which means the total communication cost is $O(\frac{k\log n}{\eps^2})$ bits. The algorithm is presented in Algorithm \ref{alg:l2_hh_static}. 
\begin{algorithm}
\caption{$\ell_2$ heavy hitter estimation in static setting}
\label{alg:l2_hh_static}
\KwIn{Accuracy parameter $\eps$}
\Comment{Sites side}
\For{Site $i \in [k]$}{
$F_2(v^{(i)}) \leftarrow  \sum_{j = 1}^{n}v^2_{ij}$.\\
\For{$j\in [n]$}{
$p_{ij} \leftarrow \min\{1, \frac{3v_{ij}^2}{\eps^2 F_2(v^{(i)})}\}$. \\
Send $v_{ij}$ to the server with probability $p_{ij}$.
}
}
\Comment{Coordinate side}
\For{$j\in [n]$}{
$\hat{v}_j \leftarrow \sum_{i = 1}^k \frac{v_{ij}}{p_{ij}}\cdot\mathbf{1}(\text{Site $i$ sends $v_{ij}$})$.
}
\end{algorithm}
\begin{theorem}
\label{thm:l2_hh_static}
For each $j\in [n]$, with probability at least $\frac{2}{3}$, the estimation given by Algorithm \ref{alg:l2_hh_static} satisfies, $|\hat{v}_j - v_j|\leq \eps \ell_2'(v)$. The expected communication cost is $O(\frac{k\log n}{\eps^2})$ bits.
\end{theorem}
\begin{remark}
	By running $\log (\frac{1}{\delta})$ independent copies of Algorithm \ref{alg:l2_hh_static} and outputting the median, the failure probability can be lowered from constant to $\delta$ (standard median trick).
\end{remark}

\subsubsection{$\ell_p$-HH in Coordinator Model}\label{sec:static-lp-hh}
Applying our $\ell_2$-HH algorithm with parameter $\eps'$, we have $|\hat{v}_j - v_j| \leq \eps' \ell_2'(v)$ for each $j \in [n]$. For $\ell_p$-HH, by Hölder's inequality, we have $\ell_2'(v) \leq (nk)^{1/2 - 1/p} \ell_p'(v)$.  So if we set $\eps' = \frac{\eps}{(nk)^{1/2-1/p}}$ in the $\ell_2$-HH algorithm, the error is $|\hat{v}_j - v_j| \leq \eps' \ell_2'(v) \leq \eps \ell_p'(v)$ as required. However, the communication cost now becomes $\tilde{O}(\frac{k}{\eps'^2}) = \tilde{O}(\frac{k \cdot (nk)^{1-2/p}}{\eps^2})$, which has a polynomial dependence on $n$. Note that the above reduction is tight since Hölder's inequality is tight in the worst case. In fact, the $F_p$ algorithm in \cite{braverman2013generalizing} used the same argument as above, and the space bound has an $n^{1-2/p}$ dependence, which is known to be necessary. That being said, it seems the $n^{1-2/p}$ factor is unavoidable, at least if we use the above reduction via Hölder's inequality.

Woodruff and Zhang~\cite{woodruff2012tight} presented the first algorithm for distributed $F_p$ with $\tilde{O}(k^{p-1}/\eps^{O(p)})$ communication. However, their algorithm and analysis are quite complex, leaving it unclear why the $n^{1-2/p}$ factor can be removed as opposed to the streaming setting. Here, our techniques and analysis are conceptually much simpler and thus provide a clearer understanding of this phenomenon.  We still apply Hölder's inequality except that we first sparsify the local frequency vectors.
Note that our goal is an additive error of $\eps \ell_p(v)$, and thus if $v_{ij} \leq \frac{\eps \ell_p(v)}{k}$, we can safely discard it. Of course $\ell_p$ is unknown; we simply use $\ell_p'$ instead, which is an underestimate of $\ell_p$ and easy to compute in the distributed setting.
This results in sparsified local frequency vectors $\tilde{v}^{(i)}$, where $\tilde{v}_{ij} = 0$ if $v_{ij} \leq \frac{\eps \ell_p'(v)}{k}$, otherwise $\tilde{v}_{ij} = v_{ij}$. A critical observation is that $\ell_2'^2(\tilde{v}) \leq (\frac{k}{\eps})^{p-2}\ell_p'^2(\tilde{v})$, which is a consequence of the follow lemma by setting $\beta = \frac{\eps}{k}$. 
\begin{lemma}\label{lem:l2-lp-inquality}
Let $x$ be a non-negative vector of dimension $d$. If $x_i\ge \beta \ell_p(x)$ for all $i\in [d]$, then
\begin{align*}
    \ell_2^2(x) \le \frac{1}{\beta^{p-2}}\ell_p^2(x).
\end{align*}
\end{lemma}
Therefore, we can run the $\ell_2$-HH algorithm on $\tilde{v}$ and set $\eps' = \frac{\eps^{p/2}}{k^{p/2-1}}$. The estimation $\hat{v}_j$ now satisfies $|\hat{v}_j - v_j| \leq |\hat{v}_j - \tilde{v}_j| + |\tilde{v}_j - v_j| \leq \eps' \ell_2'(\tilde{v}) + \eps \ell_p'(v) \leq 2\eps\ell_p'(v)$. The communication cost is $\tilde{O}(\frac{k}{\eps'^2}) = \tilde{O}(\frac{k^{p-1}}{\eps^p})$. To generate $\tilde{v}$, each site needs to know $\ell_p'(v)$, which can done using $\tilde{O}(k)$ communication. The full algorithm is presented in Algorithm \ref{alg:lp_hh_static}.
\begin{algorithm}
\caption{$\ell_p$ heavy hitter estimation in static setting (two rounds)}
\label{alg:lp_hh_static}
\Comment{Calculate $\ell_p'$ first}
\For{Site $i\in[k]$}{
$F_p(v^{(i)}) \lar \sum_{j}^n v_{ij}^p$.\\
Send $F_p(v^{(i)})$ to the coordinator.
}
Coordinator: calculate $\ell_p'(v)= (\sum_{i = 1}^k F_p(v^{(i)}))^{1/p}$ and broadcast $\ell_p'(v)$ to all sites.\\
\For{Site $i \in [k]$}{
Define $\tilde{v}^{(i)} = (\tilde{v}_{i1}, \tilde{v}_{i2},\ldots,\tilde{v}_{in})$, where $\tilde{v}_{ij} = v_{ij}$ if $v_{ij} \geq \frac{\eps\ell_p'(v)}{k}$, otherwise $\tilde{v}_{ij}  = 0$.}
Set $\eps' \lar \frac{\eps^{p/2}}{k^{p/2 -1}}$.\\
Run the $\ell_2$ heavy hitter algorithm (Algorithm \ref{alg:l2_hh_static}) with each site's frequency vector set as $\widetilde{v^{(i)}}$ and accuracy parameter as $\eps'$. Denote the estimates returned by $\hat{v}_1, \hat{v}_2, \ldots, \hat{v}_n$.
\end{algorithm}

\begin{theorem}
\label{thm:lp_hh_static}
For each $j\in [n]$, with probability at least $\frac{2}{3}$, the estimates given by Algorithm \ref{alg:lp_hh_static} satisfy $|\hat{v}_j - v_j|\leq 2\eps \ell_p'(v)$. The expected communication cost is $O(\frac{k^{p - 1}\log n}{\eps^p})$ bits.
\end{theorem}

\paragraph{One-round $\ell_p$-HH.} Note that Algorithm \ref{alg:lp_hh_static} requires one additional round of communication to calculate $\ell_p'$. However, this round can be saved through a simple trick. In Algorithm \ref{alg:lp_hh_static}, the value of $\ell_p'$ is needed to generate $\tilde{v}$. The analysis shows that a constant approximation of $\ell_p'$ is sufficient. Therefore, each site can run $O(\log n)$ instances of $\ell_p$-HH in parallel; the $i$th instance guesses $2^i$ as the value of $\ell_p'$. The coordinator computes the true $\ell'_p$ value and select the corresponding instance as the final output. This modification eliminates the additional round required to compute and broadcast $\ell_p'$ with communication increased by at most a factor of $O(\log n)$. Details is provided in the appendix.

\subsection{Heavy Hitters in Distributed Tracking Model}

\subsubsection{$\ell_2$-HH in Distributed Tracking Model}\label{sec:l2-hh-tracking}
Our algorithms for frequency moments all boils down to $\ell'_2$-HH, and thereby extending Algorithm~\ref{alg:l2_hh_static} to the tracking model is the most important technical step. In the literature, the complexity of $\ell_1$-HH in the distributed tracking model is well-understood~\cite{huang2012randomized,woodruff2012tight}, while there is no non-trivial result for $\ell_2$-HH mainly due to its non-linear nature. Existing techniques for $\ell_2$-HH are sketch-based and are thus difficult, if not impossible, to achieve optimal communication in the tracking model. On the other hand, our static algorithm for $\ell_2$-HH is sampling based, which is potentially much easier to simulate in the dynamic setting. However, there is still a key technical challenge. Existing sampling based algorithms for $\ell_1$-HH send each local (item, frequency) pair with probability linear in the local frequency. This linearity in the sampling probability is the key factor underneath the success of current $\ell_1$-HH monitoring algorithms. Unfortunately, our sampling algorithm for $\ell_2$-HH critically relies on the sampling probability being a quadratic function of the item frequency. 
In this paper, we propose a general technique for dynamizing static distributed sampling algorithms with non-linear sampling probability. At a high level, we substitute the non-linear sampling function with a piece-wise linear approximation and divide the whole monitoring period into multiple intervals each with a linear sampling function. 

\paragraph{Our algorithm.} 
The full algorithm consists of Algorithms \ref{alg:l2_hh_tracking_site} and \ref{alg:l2_hh_tracking_coordinator}. On each site $i$, the tracking algorithm runs in $O(\log n)$ rounds: the local $F_2$ value doubles in each round. We can estimate the frequency of $j$ in different rounds separately. Since the variance from each round increases geometrically (the variance of a round is proportional to the $F_2$ value at the start of that round), the total variance is bounded by the variance of the last round. Therefore, we only need to focus on a fixed round.

Let $F$ be the local $F_2$ value on site $i$ at the start of the round. For an element $j$, let $w_{ij}$ denote the increment of $v_{ij}$ during the round. Our goal is to track $w_{ij}$ with variance bounded by $\eps^2 F$.
At any fixed time $t$ in a round, according to our static $\ell_2$-HH algorithm, we only need to make sure that element $j$ is sampled with probability $p_{ij}(t) \propto \frac{w_{ij}^2(t)}{\eps^2 F}$.

The value of $w_{ij}^2$ may become larger than $\eps^2 F$ so that $p_{ij} = 1$ in the rest of the round, which requires to update the value of $w_{ij}$ continuously. To avoid this, for each $j$ we further divide each round into multiple phases\footnote{Each element $j$ has its own phase partition, but the total number of phases among all elements is bounded by $1/\eps^2$.}, each corresponding to an $\eps \sqrt{F}$ increment in $v_{ij}$. Note an $\eps \sqrt{F}$ increment in $v_{ij}$ implies 
an increment of at least $\eps^2F$ in $v^2_{ij}$, and hence in the local $F_2$. This mean there are at most $1/\eps^2$ phases finish during a round over all elements $j$.
Site $i$ will send the current $v_{ij}$ to the coordinator when a phase of $j$ ends. 
So, the error of estimating $w_{ij}$ only comes from the last unfinished phase. 

Suppose the unfinished phase start from time $t_0$. We wish to estimate the increment $w_t := w_{ij}(t) - w_{ij}(t_0)$ within variance $\eps^2 F$, where $t$ is the current time. By definition, $w_t\in [0,\eps \sqrt{F}]$. Consider the following sampling scheme: at the beginning of the phase, draw a random number $r\in _R [0,\eps \sqrt{F}]$; send a message to the coordinator once $w_t$ reaches $r$; the coordinator sets $\hat{w}_t = \eps \sqrt{F}$ if a message is received and otherwise $0$. At any fixed time $t$ during the phase, the probability that a message is sent to the coordinator is $\frac{w_t}{\eps \sqrt{F}}$, and thus $\hat{w}_t$ is an unbiased estimate. This linear sampling probability does not meet our requirement. One may argue that we can draw $r\in_R [0,\eps^2 F]$, and send a message whenever $w_t^2 = r$. Then the sampling probability is $\frac{w_t^2}{\eps^2 F}$ as required. However, the problem is that in this case, the natural unbiased estimate is $\frac{\eps^2 F}{w_t}$, which depends on $w_t$. This is not an issue in the static case since we can sent the exact value of $w_t$ if it is sampled. In the tracking model, $w_t$ keeps changing after it is sampled. Next we provide a method, using the above linear sampling procedure as a building block, to approximately implements a quadratic sampling scheme.

The idea is as follows. We further divide a phase into intervals, where $w_t$ doubles in each interval. More specifically, in the $c$th interval, $w_t$ increases from $2^{c}$ to $2^{c+1}$. At the beginning of this interval, we draw a random number $r_c \in_R [0, \frac{\eps^2 F}{2^{c}}]$, and site $i$ sends a message to the coordinator the first time when $w_t -2^c = r_c$. Define $\hat{w}_c = \mathbf{1} (\textrm{coordinate receives a message in the $c$th interval})\cdot \frac{\eps^2 F}{2^{c}}$. Then the coordinate estimates $w_t$ by  $\hat{w}_t = \sum_{c=1}^{\log (\eps \sqrt{F})} \hat{w}_c.$
It is easy to see $\hat{w}_t$ is an unbiased estimate of $w_t$. For the variance, if interval $c$ has not started, then $\var[\hat{w}_c]=0$; otherwise $\var[\hat{w}_c]\le \eps^2 F$ by an easy calculation. As there are at most $\log(\eps \sqrt{F})$ intervals, the total variance is at most $\log(\eps \sqrt{F})\cdot \eps^2 F$. To sum up, for element $j$ on site $i$, since the estimation error only come from the last phase, whose variance is bounded by $\log(\eps \sqrt{F})\cdot \eps^2 F$, we get an unbiased estimate of $w_{ij}$ (the increment of $v_{ij}$ in the current round) with desired variance (after scaling $\eps$ by a $\log$ factor). For the communication cost of a round, note there are at most $1/\eps^2$ phases finish during a round, each contributing at most one message, so the cost corresponding to 
all finished phases is $\tilde{O}(k/\eps^2)$. Additionally, each element $j$ has one unfinished phase on each site, and the total number of such phases is $nk$. According to our sampling scheme, for each $i,j$, the probability that site has sent the message corresponding to the current phase is bounded by $O\left(\frac{v_{ij}^2}{\eps^2 F}\right)$. This is the same quadratic sampling probability as in our static $\ell_2$-HH algorithm, which means the expected cost of this part is $\tilde{O}(k/\eps^2)$.

\begin{theorem}
\label{thm:l2_hh_tracking}
For any time $t\in [m]$, for any $j \in [n]$, with constant probability, our $\ell_2$-HH tracking algorithm produces estimate satisfying $|v_{j}(t) - \hat{v}_j(t)|\leq \eps \ell_2'(t)$. The expected communication cost is $O(\frac{k}{\eps^2}\cdot \log^3 n)$.
\end{theorem}

\begin{algorithm}
\caption{$\ell_2$ heavy hitter tracking: site $i$, tracking $j$}
\label{alg:l2_hh_tracking_site}
\Comment{$d,s,c$:current round, phase, and interval. $w$ is increment of $v_{ij}$ in the current phase, $\Delta$ is the increment of $v_{ij}$ in the current interval.}
$d=1, s=1, c=1, w=0, F=0$\\
$r_{s,c}\sim \mathsf{Uni}(0, \frac{\eps^2F}{2^c})$, $\Delta = 0$\\
\For{$t=1,2,\cdots$}{
    \If{Site $i$ receives element $j$}
    {
    Update $v_{ij}(t) = v_{ij}(t-1)+1, w=w+1, \Delta=\Delta+1$\\
    $F_2 = F_2(v^{(i)}(t))$\\
    \If{$\Delta == r_{s, c}$}{
        Send $(0, i, j, \frac{\eps^2F}{2^c})$ to the coordinator.
    }
    \If{$\Delta \ge 2^c$}{$ c=c+1, \Delta = 0$ \Comment*[r]{start a new interval}
    $r_{s,c}\sim \mathsf{Uni}(0, \frac{\eps^2F}{2^c})$} 
    \If{$w \ge \eps\sqrt{F}$}{$s = s+1, c=1, w=0, \Delta = 0$ \Comment*[r]{start a new phase}
    Send $(1, i, j, v_{ij}(t))$ to the coordinator. \label{line:exact}
    } 
    \If{$F_2\ge 2F$}
    {$d=d+1, s=1, c=1, w=0, \Delta = 0$  \Comment*[r]{start a new round}
     $F = F_2$
    }
    }
}
\end{algorithm}

\begin{algorithm}
\caption{$\ell_2$ heavy hitter tracking, coordinator}
\label{alg:l2_hh_tracking_coordinator}
Initialize $\hat{v}_{ij} = 0$ for $i\in[k], j\in[n]$.\\
\If{the coordinator receives a message $(m, i, j, v)$}{
\eIf{m==0}{$\hat{v}_{ij} = \hat{v}_{ij} + v$}{$\hat{v}_{ij} = v$}
}
Output $\hat{v}_j = \sum_{i=1}^k \hat{v}_{ij}$ for $j\in [n]$.
\end{algorithm}

\subsubsection{$\ell_p$-HH in Distributed Tracking Model}

Similar as that in static setting, $\ell_p$ heavy hitter algorithm can be obtained by simple modification of $\ell_2$ heavy hitter tracking, see Algorithm \ref{alg:lp_hh_tracking}. Its accuracy and communication guarantee is stated in the following theorem. 
\begin{theorem}
\label{thm:lp_hh_tracking}
For any $t$, for any $j \in [n]$, with constant probability, the estimation given by the $\ell_p$-HH tracking algorithm satisfies $|v_{j}(t) - \hat{v}_j(t)|\leq \eps \ell_p'(t)$. The expected communication cost is $O(\frac{k^{p-1}}{\eps^p}\cdot \log^4 n)$.
\end{theorem}

\section{Distributed Frequency Moments Estimation}
\subsection{$F_p$ in Coordinator Model}
With our $\ell_p$-HH algorithms, we can solve $F_p$ quite easily. Note that the key of recursive sketching is computing $(\alpha, \eps)$-covers. Let $v$ be the frequency vector of the union of $k$ distributed sets. Let $v^p$ denote the vector whose entries are $p$th power of entries in $v$. We apply recursive sketching on $v^p$. Recall that an $(\alpha, \eps)$-cover of $v^p$ is a set including all $\alpha$-$\ell_1$-HH of $v^p$ (i.e.\ $\alpha^{1/p}$-$\ell_p$-HH of $v$), and for each element in the cover its value can be $(1 \pm \eps)$ approximated (see Definition \ref{def:hh_cover}). 
In our two-round protocol for $(\alpha,\eps)$-covers, the coordinator first computes a set of indices which are $\alpha^{1/p}$-$\ell_p$-HH of $v$ using our one-round $\ell_p$-HH algorithm (Algorithm~\ref{alg:1round_lp_hh_static}). Note this set has size bounded by $O(1/\alpha)$. The coordinator then broadcast this set and all sites send back the true local counts of those elements. In this way, the coordinator obtains an $(\alpha, 0)$-cover of $v^p$ using two rounds of communication. The communication cost is $\Tilde{O}\left(\frac{k^{p-1}}{\alpha} \right)$. In our one-round protocol of $(\alpha,\eps)$-cover, we simply apply our one-round $\ell_p$-HH algorithm setting the error parameter $\eps' = \alpha^{1/p}\eps/O(p)$. 
\begin{algorithm}
\caption{Two-round distributed $(\alpha, 0)$-cover}
\label{alg:2round_hh_cover}
\FirstRound{
Run the one-round $\ell_p$ heavy hitter algorithm with $\eps = \frac{\alpha^{1/p}}{4}$ and error probability $\delta=1/n^2$. \\
Select the top $(\frac{4^p}{\alpha})$-largest $\hat{v}_j$, denoted as $S'$.
Broadcast $S'$ to all sites.
}
\SecondRound{
\For{$j \in S'$}{
For each site $i \in [k]$, send $v_{ij}$ to the coordinator so that the coordinator can calculate $v_j$. \\
}
$Q \lar \{(j, v_j^p) | j \in S'\}$.\\
Return $Q$.
}
\end{algorithm}

\begin{algorithm}
\caption{One-round distributed $(\alpha, \eps)$-cover}
\label{alg:1round_hh_cover}
Set $\alpha' \lar \alpha^{1/p}$.\\
Run the one-round $\ell_p$ heavy hitter algorithm with $\eps' = \frac{\alpha' \eps}{O(p)}$ and  error probability $\delta=1/n^2$. \\
Select the top $(\frac{4^p}{\alpha})$-largest $\hat{v}_j$, denoted as $S'$.\\
$Q \lar \{(j, \hat{v}_j^p) | j \in S'\}$\\
Return $Q$.
\end{algorithm}

\begin{theorem}
\label{thm:hh_cover}
    With probability $1-\frac{1}{n}$, Algorithm \ref{alg:2round_hh_cover} returns an $(\alpha, 0)$-cover and Algorithm \ref{alg:1round_hh_cover} returns an $(\alpha, \eps)$-cover of $v^p$ of size $O(\frac{1}{\alpha})$ with communication cost $O(\frac{k^{p - 1}}{\alpha}\cdot \log^3 n)$ and $O(\frac{k^{p -1}}{\alpha\eps^p}\cdot \log^3n)$.
\end{theorem}
\begin{proof}
    For Algorithm \ref{alg:2round_hh_cover}, by the accuracy guarantee of $\ell_p$ heavy hitter algorithm, we have $|\hat{v}_j - v_j| \leq \eps \ell_p' = \frac{\alpha^{1/p}}{4}\cdot \ell_p$. Set $\alpha' = \alpha^{1/p}$. Define $S_1 := \{j\in[n] |\hat{v}_j > \frac{\alpha'}{2}\ell_p \}$ and $S_2 := \{j\in[n] |{v}_j > \frac{\alpha'}{4}\ell_p \}$. For $\forall j \in S_1$, 
    \begin{align*}
        \hat{v}_j > \frac{\alpha'}{2}\ell_p \Rar v_j + \frac{\alpha'}{4}\ell_p  > \frac{\alpha'}{2}\ell_p  \Rar v_j > \frac{\alpha'}{4}\ell_p  \Rar j \in S_2,
    \end{align*}
    which implies $|S_1| < |S_2| < \frac{4^p}{\alpha'^p} = \frac{4^p}{\alpha}$. Therefore, $S_1 \subseteq S'$. On the other hand, for any $v_j > \alpha' \ell_p$ (i.e., $v_j^p > \alpha F_p$), 
    \begin{align*}
        v_j > \alpha' \ell_p \Rar \hat{v}_j + \frac{\alpha'}{4}\ell_p > \alpha' \ell_p \Rar \hat{v}_j >= \frac{3}{4}\alpha'\ell_p \Rar j \in S_1  \Rar j \in S',
    \end{align*}
    which means $S'$ satisfies the first condition of $(\alpha, 0)$-cover (see Definition \ref{def:hh_cover}). In the second round, for any $j \in S'$, $v_j$ can be exactly calculated, thus the second condition of $(\alpha, 0)$-cover for $Q$ satisfies directly. Now consider the communication. The $\ell_p$ heavy hitter algorithm needs $O(\frac{k^{p - 1}}{\eps^p}\log^3n) = O(\frac{4^pk^{p - 1}}{\alpha}\log^3 n)$. Since $|S'| = O(\frac{4^p}{\alpha})$, the communication cost of broadcasting $S'$ and calculating $v_j$ for $j \in S$ is $O(\frac{4^pk}{\alpha}\cdot \log n)$. Thus the total communication cost of Algorithm \ref{alg:2round_hh_cover} is $O(\frac{4^pk^{p -1}}{\alpha}\log^2 n)$.
    
    For the one round protocol, using the same argument as above, if $v_j > \alpha^{1/p}\ell_p = \alpha'\ell_p$, we must have $(j, \hat{v}_j) \in Q$. Thus $Q$ satisfies the first condition of $(\alpha, \eps)$-cover. At the same time, for $(j, \hat{v}_j) \in Q$, we have $v_j > \hat{v}_j - \frac{\alpha'}{4}\ell_p' > \frac{\alpha'}{4}\ell_p'$. Then
    \begin{align*}
        |\hat{v}_j - v_j| \leq \frac{\alpha'\eps}{O(p)} \ell_p' \leq \frac{\eps}{O(p)} v_j. 
    \end{align*}
    Consequently, 
    $$\hat{v}_j^p = (1\pm \frac{\eps}{O(p)})^p v_j^p = (1\pm \eps) v_j^p.$$
    Therefore, $Q$ is a $(\alpha, \eps)$-cover. The communication complexity of the one-round protocol is $O\left(\frac{p^{O(p)} k^{p -1}}{\eps'^{p}}\log^3 n\right)$, which completes the proof.
\end{proof}

Applying our distributed $(\alpha,\eps)$-cover algorithms, we obtained distributed algorithms for $F_p$ estimation via Theorem~\ref{thm:recursive_sketch}. In the two-round algorithm, we compute $(\frac{\eps^2}{\log^3n},0)$-covers using Algorithm~\ref{alg:2round_hh_cover}; in the one-round algorithm we compute $(\frac{\eps^2}{\log^3n},\eps)$-covers using Algorithm~\ref{alg:1round_hh_cover}.
\begin{theorem}[Static distributed $F_p$ estimation]
\label{thm:fp_static}
For $F_p$ estimation in the coordinator model,
\begin{enumerate}
    \item There exists a 2-round protocol that uses $O\left(\frac{k^{p - 1}}{\eps^2}\cdot \polylog(n)\right)$ communication bits.
    \item There exists a 1-round protocol that uses $O\left(\frac{k^{p - 1}}{\eps^{p+2}}\cdot \polylog(n)\right)$ communication bits.
\end{enumerate}
\end{theorem}

\subsection{$F_p$ in Distributed Tracking Model}
\label{sec:fp_tracking}
 In this section, we present an $F_p$ tracking algorithm based on the $\ell_p$-HH tracking algorithm (Theorem~\ref{thm:l2_hh_tracking} and \ref{thm:lp_hh_tracking}). The communication cost is $O(\frac{k^{p-1}}{\eps^2}\polylog(n))$ bits which is near optimal~\cite{woodruff2012tight}. We emphasize that, even for $p=2$, our algorithm is the first near optimal algorithm improving upon \cite{woodruff2012tight,cormode2008algorithms}

Our tracking algorithm consists of multiple rounds. We will see that after a round ends, the $F_p$ value increases by a constant factor, and thus there are $O(\log n)$ rounds in total. In the beginning of each round, we will use our static $F_p$ algorithm to obtain an $(1\pm \eps)$-approximation of current $F_p$, denoted as $\hat{F}$. At the same time, we run Algorithms~\ref{alg:2round_hh_cover} obtain an $(\eps^2, 0)$-cover of $v^p$ with size no more than $O(\frac{1}{\eps^2})$. Within this round, our goal is to maintain an $(\eps^2, \eps)$-cover $Q$ of $v^p(t)$ at all times. By our $\ell_p$-HH tracking algorithms, it is easy to ensure all heavy hitters are collected in $Q$. For each element $j\in Q$, we need to track an $(1\pm \eps)$-approximations of $v_j(t)$. Although the size of $Q$ in our algorithm never exceeds $O(1/\eps^2)$, if we track each element's frequency independently, the cost will be at least $\Omega((k+\frac{\sqrt{k}}{\eps})|Q|)$~\cite{huang2012randomized}. 

To overcome this issue, we observe that from the analysis of recursive sketching in~\cite{braverman2013generalizing}, point-wise $\eps$-relative error is not necessary, and an $\eps$ relative error approximation of their sum is sufficient. Thus, we make the following relaxation.
\begin{definition}[weak ($\alpha, \eps$)-cover]
\label{def:weak_hh_cover}
A non-empty set of (index, value) pairs $Q = \{(i_1, w_1),\ldots, (i_t, w_t)\}$  is called a weak $(\alpha, \eps)$-cover w.r.t.\  $u$, if the following two conditions hold:
\begin{enumerate}
    \item  $Q$ contains all heavy hitters: $\forall i \in [n]$, if $u_i > \alpha \sum_{i}u_i$, then $i \in Q$.
    \item  Sum approximation: $\left|\sum_{(j,w)\in Q} w - u_j\right|  \leq \eps |u|$.
\end{enumerate}
\end{definition}
 More specifically, in the error analysis of recursive sketching in \cite{braverman2013generalizing}, for each $l$, we need to bound $Err_l^3$:
\begin{align*}
    Err_l^3 =\left|\sum_{(j,w) \in Q_l}(1-2h_{l + 1,j})\left(w - u_{l,j}\right)\right|,
\end{align*}
The analysis requires $Err_l^3 \leq \eps |u_l|$, and thus $\eps$ relative error for each $u_{j}$ is sufficient (but not necessary). Define $Q_{l,b} =\{(j,w)~|~(j,w)\in Q_l, h_{l+1,j} =b \}$ for $b=0,1$. We observe that
\begin{align*}
    Err_l^3 \leq \left|\sum_{(j,w) \in Q_{l,0}}\left(w - u_{l,j}\right)\right| + \left|\sum_{(j,w)\in Q_{l,1}}\left(w - u_{l,j}\right)\right|.
\end{align*}
We use $\odot$ to denote entry-wise product. So if $Q_{l,1}$ is a weak $(\alpha,\eps)$-cover for $u^{l,1}:=u^{l}\odot h_{l+1}$ and $Q_{l,0}$ is a weak $(\alpha,\eps)$-cover for $u^{l,0}:=u^{l}\odot \bar{h}_{l+1}$, $Err_l^3 \le \eps|u_l|$ as required. Thus, in Algorithm~\ref{alg:recursive_sktech}, it is sufficient to compute a weak $(\frac{\eps^2}{\phi^3},\eps)$-cover w.r.t.\ $u^{l,b}$ for $l=1,\cdots,\phi$ and $b=0,1$. 

Now our problem reduces to tracking a weak $(\alpha, \eps)$ cover w.r.t.\ $v^p$, where $v$ is the frequency vector of $k$ distributed streams. The indices of the heavy hitters can be maintained using our $\ell_p$-HH tracking algorithm. It remains to track the frequency of each element registered in $Q$. Since we only need to approximate their sum within additive error $\eps |v^p|$, it is sufficient to estimate each $v_j^p$ within variance $\eps^2 |v^p| v_j^p$. To this end, we use the same thresholding sampling scheme as in Algorithm~\ref{alg:l2_hh_tracking_site}: pick a random threshold $r\in_R [0,{\eps^2 |v^p|}]$, the estimate $\hat{u}_j = \eps^2 F$ if $v_j^p\ge r$ otherwise $0$. The threshold can be generated by public randomness. The correctly set the value of $\hat{u}_j$, the coordinator needs to be notified immediately once $v_j^p \ge r$ or $v_j \ge r^{1/p}$. This is the classic thresholded sum tracking problem, which can be solved with $\Tilde{O}(k)$ bits of communication.

\begin{algorithm}
\caption{Tracking a weak $(\alpha,\eps)$-cover of $v^p$}
\label{alg:fp_tracking}
Start an instance of Algorithm \ref{alg:lp_hh_tracking} with accuracy parameter set as $\alpha' = \alpha^{1/p}/4$, denoted by $\alg_{hh}$. \\
Start a sum tracking algorithm to track a $0.1$-approximation of $F_p'$, denoted as $\hat{F}_p'(t)$. \Comment*[r]{Recall $F_p'$ is the sum of local $F_p$'s, so tracking $F_p'$ is a standard sum tracking problem, which can be solved trivial with $\Tilde{O}(k)$ communication \cite{yi2009optimal}}

\StartRound{
Let $t_0$ be the start time of current round.\\
Run the $2$-round $F_p$ protocol and get a small constant approximation of $F_p(t_0)$ denoted by $\hat{F}$. \\
Run Algorithm~\ref{alg:2round_hh_cover} and get $Q(t_0)$, which is an $(\frac{\alpha}{6^{p + 3}}, 0)$-cover at $t_0$. \\
$I(t_0) \lar Ind(Q(t_0))$. \\
\For{$j \in I(t_0)$}{
Start an instance of Algorithm \ref{alg:vj_tracking_hh} to track $v_j^p$; let $w_j(t)$ denote the output.
}
}
Let $N_r(t)$ be the total number of phases that have been completed across all instances of Algorithm \ref{alg:vj_tracking_hh} in the current round. \\
\If{$N_r(t) > \frac{3}{\eps^2}$
}{
End current round and start next round.
}
\If{$\hat{F}_p'(t) > 4\hat{F}$}{
End current round and start next round.
}
Let $\hat{v}_j(t)$ be the estimation of $v_j(t)$ given by $\alg_{hh}$.  \\
$I'(t) \lar \{j | \hat{v}_j \geq \frac{2}{3}{\alpha^{1/p}}\hat{F}^{1/p}\}$\\
$I(t) \lar I'(t) \cup I(t-1)$\\
\If{$|I(t)| - |I(t_0)| > \frac{6^{p + 3}}{\alpha}$
}{
End current round and start next round.
}
\For{$j \in I(t)-I(t - 1) $}{
Start an instance of Algorithm \ref{alg:vj_tracking_hh} to track $v_j^p$; let $w_j(t)$ denote the output. \\
}
$Q(t) \lar \{(j, w_j(t)) | j\in I(t)\}$.
\end{algorithm}

\begin{algorithm}
\caption{Tracking $v_j^p$}

\label{alg:vj_tracking_hh}
The tracking process consists of multiple phases. \\
\StartPhase{
Let $t_s$ be the start time of current phase.
All sites and the coordinator use public randomness to generate a random number $r \sim \mathsf{Uni}([0, \eps^2 \hat{F}])$. \\
For each site $i \in [k]$, send $v_{ij}$ to the coordinator to compute $v_j(t_s)$ \\
Initialize $w_j(t) = v^p_j(t_s)$\\
Start an instance of exact thresholded sum tracking algorithm of~\cite{cormode2008algorithms} with threshold $(r + v_j^p(t_s))^{1/p}$; denote this instance as $\alg_1$ \Comment*[r]{tracking $\sum_{i=1}^k v_{i,j}$}
Start another instance of exact thresholded sum tracking algorithm with threshold $(\eps^2 \hat{F} +  v_j^p(t_s))^{1/p}$; denote this instance as $\alg_2$.
}

Let $t$ be the current time\\
\If{$\alg_1$ outputs $1$}{
$w_j(t) \lar v_j^p(t_s) + \eps^2 \hat{F}$\\
}

\When{$\alg_2$ outputs $1$ }{
End current phase and start next phase.
}
\end{algorithm}

\begin{lemma}
	\label{lem:vj_tracking_hh}
	In Algorithm \ref{alg:vj_tracking_hh}, at any time $t$, we have
	\begin{enumerate}
		\item $\E[w_j(t)] =v_j^p$, $\var[w_j(t)] \leq \eps^2 \hat{F} v_j^p$;
		\item After one phase, $v_j^p$ increases by $\eps^2\hat{F}$;
		\item The communication cost of one phase is $\Tilde{O}(k)$.
	\end{enumerate}
\end{lemma}
\begin{proof}
	In the beginning of a phase $w_j(t) = v_j^p$. For a fixed time $t$ during a phase, 
	$$w_j(t) = \mathbf{1}(\alg_1 \textrm{ output 1}) \cdot \eps^2\hat{F} + v_j^p(t_s).$$
	Note
	\begin{align*}
		\Pr[\alg_1 \textrm{ output 1}] = \Pr[v_j(t) \ge (r+v_j^p(t_s))^{1/p}] = \frac{v_j^p(t) - v_j^p(t_0)}{\eps^2 \hat{F}},
	\end{align*}
	which implies $\E[w_j(t)] = v_j^p(t)$. The variance of $w_j(t)$ is
	\begin{align*}
		\var[w_j(t)] \le \eps^4 \hat{F}^2 \cdot \frac{v_j^p(t) - v_j^p(t_0)}{\eps^2 \hat{F}} \le \eps^2 \hat{F} v_j^p(t).
	\end{align*}
	2 is trivial. 3 holds because the exact thresholded sum tracking algorithm from~\cite{cormode2008algorithms} has communication $\Tilde{O}(k)$.
\end{proof}

\begin{theorem}
	\label{thm:cover_tracking}
	For any $t\in[m]$, with probability at least $\frac{7}{10}$, the set $Q(t)$ from Algorithm \ref{alg:fp_tracking} is a weak $(\alpha,\eps)$-cover of $u=v^p$, where $v$ is the frequency vector of $k$ distributed streams. The communication cost is $\Tilde{O}\left(\frac{k^{p-1}}{\alpha} +\frac{k}{\eps^2} \right)$.
\end{theorem}
\begin{proof}
	We only need to focus on the analysis of a fixed round whose duration is from time $t_0$ to $t_e$. At time $t_0$, $Q(t_0)$ is an $(\alpha/2^{p+3}, 0)$-cover. For any $t\in [t_0, t_e]$, if $v_j^p(t) \geq \alpha F_p(t)$, then   
	$$
	\hat{v}_j(t) > v_j(t) - \frac{\alpha^{1/p}}{4}\ell'_p(t) \geq \frac{3}{4}\alpha^{1/p}\ell_p(t) \geq \frac{3}{4}\alpha^{1/p}\ell_p(t_0) \geq \frac{2}{3}\eps^{2/p}\hat{F}^{1/\alpha}.
	$$ 
	Thus $j$ will be included in the set $I(t)$. Hence $Q(t)$ satisfies the first condition of cover set.
	
	For each $j\in Q(t)$, we run an instance of Algorithm~\ref{alg:vj_tracking_hh} which outputs an estimate $w_j(t)$ of $v_j^p(t)$. By Lemma~\ref{lem:vj_tracking_hh}, $w_j(t)$ is unbiased and has variance at most $\eps^2 F v_j^p(t)$. Then the total variance  
	$$\var[\sum_{j\in Q(t)} w_j(t)] \le O(\eps^2F^2).$$
	By Chebyshev's inequality, the second condition of weak cover set holds with constant probability.
	
	\paragraph{Communication cost.} $\alg_{hh}$ (line 1) incurs $\Tilde{O}\left(\frac{k^{p-1}}{\alpha}\right)$ bits of communication. The sum tracking (line 2) incurs $\Tilde{O}(k)$ bits. In the beginning of the round, it costs $\Tilde{O}\left(\frac{k^{p-1}}{\alpha}\right)$ bits to get a constant approximation of $F_p(t_0)$ and an $(\alpha,0)$-cover $Q(t_0)$.
	The rest of the communication costs are from instances of Algorithm~\ref{alg:vj_tracking_hh}. From Theorem~\ref{thm:hh_cover}, the size of $Q(t_0)$, i.e., $|I(t_0)|$, is bounded by $O\left(\frac{1}{\alpha}\right)$. On the other hand, in a round we always have $|I(t)| - |I(t_0)| \le O\left(\frac{1}{\alpha}\right)$, meaning $|I(t)| \le O\left(\frac{1}{\alpha}\right)$. Therefore, the total number of instances of Algorithm~\ref{alg:vj_tracking_hh} is at most $O\left(\frac{1}{\alpha}\right)$ at any time during a round. From our algorithm, the total number of completed phases over all instances is $N_r(t) \le \frac{3}{\eps^2}$ during a round. Each completed phase incurs $\Tilde{O}(k)$ communication (costs of $\alg_1$ and $\alg_2$ according to \cite{cormode2008algorithms}), and overall $\Tilde{O}\left(\frac{k}{\eps^2}\right)$. There is one ongoing phase for each $j\in Q(t)$, each of which incurs $\Tilde{O}(k)$ bits, and thus $\Tilde{O}\left(\frac{k}{\alpha}\right)$ bits in total. To sum up, the total communication of a round is $\Tilde{O}\left(\frac{k^{p-1}}{\alpha} +\frac{k}{\eps^2} \right)$.
	
	Next we bound the number of rounds.  Note that the current round ends when any of the following three conditions is satisfied:
	\begin{enumerate}
		\item $\hat{F_p'}(t) > 4\hat{F}$;
		\item $N_r(t) > \frac{3}{\eps^2}$;
		\item $|I(t)| - |I(t_0)|  > \frac{6^{p + 3}}{\alpha}$.
	\end{enumerate}
	When the first condition is satisfied, we have that
	\begin{align*}
		F_p(t) > F_p'(t) > 0.9 \hat{F_p'}(t) \geq 3 \hat{F} \geq 2 F_p(t_0),
	\end{align*}
	which means $F_p$ increases by at least a factor of $2$. For the second condition, when a phase is completed for some element $j\in Q(t)$, comparing to the start of the phase $v_j^p$ increases by $\eps^2 \hat{F}$ ($\alg_2$), which means $F_p$ increases by $\eps \hat{F} \ge 0.9\eps^2 F_p(t_0)$. So when $N_r(t) > \frac{3}{\eps^2}$ is satisfied, $F_p$ increases by at least a factor of $2$.
	
	Now consider the last condition. For any element $j$ that is newly added to $I(t)$ (i.e., originally not in $I(t_0)$), 
	\begin{align*}
		v_j(t) \stackrel{(a)}{\geq} \hat{v}_j(t) - \frac{\alpha^{1/p}}{4}\ell_p'(t) \stackrel{(b)}{\geq}  \frac{2}{3}\alpha^{1/p}\hat{F}^{1/p} - \frac{\alpha^{1/p}}{4}\ell_p'(t) \stackrel{(c)}{\geq}  \frac{\alpha^{1/p}}{6}\hat{F}^{1/p},
	\end{align*}
	where $(a)$ is by the accuracy guarantee of our $\ell_p$-HH tracking algorithm, $(b)$ is from the definition of $I(t)$ and $(c)$ is by $\ell_p'(t) \leq (4\hat{F})^{1/p}$ due to condition 1. Then, 
	\begin{align*}
		v_j^p(t) \geq \frac{\alpha}{6^{p}}\hat{F} \geq \frac{\alpha}{6^{p+1}}F_p(t_0).
	\end{align*}
	On the other hand, the fact that $j \notin I(t_0)$ implies that
	\begin{align*}
		v_j^p \leq \frac{\alpha}{6^{p + 3}}F_p(t_0),
	\end{align*}
	which means $v_j^p$ increases by at least $\frac{3}{4}\cdot\frac{\alpha F}{6^{p + 1}}$. Therefore, when the last condition meets, $F_p$ increases by at least $\frac{6^{p + 3}}{\alpha}\cdot \frac{3}{4}\frac{\alpha F}{6^{p + 1}} > 10F$. 
	
	To conclude, after a round ends, $F_p$ increases by a constant factor. Thus, the total number of rounds is $O(\log n)$.
\end{proof}

As discussed above, to solve the $F_p$ tracking problem, we only need to track $O(\log n)$ weak $(\frac{\eps^2}{\log^3 n}, \eps)$-covers in parallel, each corresponding to a sub-sample of the universe $[n]$. By Theorem~\ref{thm:cover_tracking}, the total cost is $\Tilde{O}\left(\frac{k^{p-1}}{\eps^2}\right)$.
\begin{theorem}
	\label{thm:fp_tracking}
	There is an tracking algorithm that, at any time $t$, with constant probability, outputs a $(1\pm \eps)$ approximation of $F_p(t)$. The communication cost is $O\left(\frac{k^{p-1}}{\eps^2}\cdot \polylog(n)\right)$ bits.
\end{theorem}

    

\section{Acknowledgments}
This research was supported in part by National Science and Technology Major Project (2022ZD0114802) and National Natural Science Foundation of China No.\ 62276066, No.\ U2241212.
\newpage

\bibliographystyle{abbrv}
\bibliography{main}
\appendix
\section{Appendix for Section~\ref{sec:HH}}

\begin{proof}[Proof of Theorem~\ref{thm:l2_hh_static}]
	Fix some $j \in [n]$. Firstly, it can be seen that $\hat{v}_j$ is an unbiased estimator of $v_j$:
	\begin{align*}
		\E[\hat{v}_j] &= \E\left[\sum_{i = k}\frac{v_{ij}}{p_{ij}}\cdot \mathbf{1}(\text{Site $i$ sends $v_{ij}$})\right]\\
        &=\sum_{i = k}\frac{v_{ij}}{p_{ij}}\cdot \E\left[\mathbf{1}(\text{Site $i$ sends $v_{ij}$})\right] =\sum_{i=1}^k v_{ij} = v_j.
	\end{align*}
	Next, we bound the variance of $\hat{v}_j$.
	\begin{align*}
		\var[\hat{v}_j] &= \var\left[\sum_{i = k}\frac{v_{ij}}{p_{ij}}\cdot \mathbf{1}(\text{Site $i$ sends $v_{ij}$})\right] \\&=\sum_{i = 1}^k \frac{v_{ij}^2}{p_{ij}^2}\cdot \var\left[\mathbf{1}(\text{Site $i$ sends $v_{ij}$})\right] \\
		&=\sum_{i = 1}^kv_{ij}^2\cdot  \frac{1 - p_{ij}}{p_{ij}} \leq \sum_{i = 1}^k v_{ij}^2\cdot \frac{2\eps^2 F_2(v^{(i)})}{3v_{ij}^2} = \frac{1}{3}\eps^2 F_2'(v),
	\end{align*}
	where the first inequality is from the definition of $p_{ij}$. The accuracy guarantee can be proved by standard Chebyshev's inequality. Now we consider the communication cost. Note that for each site $i\in[k]$, its expected communication cost is 
	\begin{align*}
		\sum_{j = 1}^n\log n\cdot  \E[\mathbf{1}(\text{Site $i$ send $v_{ij}$})] \leq \log n\cdot \sum_{j = 1}^n  \frac{3v_{ij}^2}{\eps^2 F_2(v^{(i)})} = {3}\cdot \frac{\log n}{\eps^2},
	\end{align*}
	where $\log n$ is the number of bits for sending $v_{ij}$. Since there are $k$ sites, the total communication cost is $O(\frac{k\log n}{\eps^2})$ bits, which completes the proof.
\end{proof}

\begin{proof}[Proof of Theorem~\ref{thm:lp_hh_static}]
	Fix an element $j \in [n]$ first. Define $\tilde{v}_k = \sum_{j = 1}^n \tilde{v}_{ij}$. By the definition of $\tilde{v}_{ij}$, we have $|v_j - \tilde{v}_{j}| \leq \eps \ell_p'$. Thus it is sufficient to show $|\hat{v}_j -\tilde{v}_{j}| \leq \eps\ell_p'(v)$. To see this, by Theorem \ref{thm:l2_hh_static}, we have
	\begin{align*}
		|\hat{v}_j - \tilde{v}_j| \leq \eps' \ell_2'(\tilde{v}).
	\end{align*}
	Note that
	\begin{align*}
		\ell_2'^2(\tilde{v}) \leq 
		\max \sum_{v_{ij}\geq\frac{\eps\ell_p'(v)}{k}} v_{ij}^2 \leq \frac{k^p}{\eps^p}\cdot (\frac{\eps\ell_p'(v)}{k})^2 = \frac{k^{p - 2}}{\eps^{p-2}}\cdot \ell_p'^2(v),
	\end{align*}
	where the inequality is from the fact that $\sum_{v_{ij}\geq\frac{\eps\ell_p'(v)}{k}} v_{ij}^2$ is maximized subject to $\sum v_{ij}^p = \ell_p'^p(v)$ when all $v_{ij}$ equals to $\frac{\eps\ell_p'(v)}{k}$. Hence we have $\eps' = \frac{\eps^{p/2}}{k^{p/2 -1}} \leq \eps {\ell_p'(v)}/{{\ell_2'(\tilde{v})}}$ which means $ |\hat{v}_j - \tilde{v}_j| \leq \eps' {\ell_2'(\tilde{v})} \leq \eps \ell_p'(v)$. By triangle inequality, it is known that $|v_j -\hat{v}_j| \leq 2\eps \ell_p'(v)$ and the accuracy guarantee is proved. Since the communication complexity of $\ell_2$ heavy hitter algorithm is $O(\frac{k}{\eps^2})$ for parameter $\eps$, the communication cost here is $O(\frac{k\log n}{(\eps')^2}) = O(\frac{k^{p-1}\log n}{\eps^p})$. On the other hand, it costs $O(k\log n)$ bits to calculate and broadcast $\ell_p'$. Thus the total communication cost is $O(\frac{k^{p-1}\log n}{\eps^p})$, which completes the proof.
\end{proof}

\begin{algorithm}
	\caption{One-round $\ell_p$ heavy hitter in static setting}
	\label{alg:1round_lp_hh_static}
	\For{$\tau \in \{1, 2, 4, \ldots, 2^{\lceil \log m \rceil}\}$}{
		\For{Site $i \in [k]$}{
			Define $\tilde{v}^{(i, \tau)} = (\tilde{v}_{i1}^{(\tau)}, \tilde{v}_{i2}^{(\tau)},\ldots,\tilde{v}_{in}^{(\tau)})$, where $\tilde{v}_{ij}^{(\tau)} = v_{ij}$ if $v_{ij} \geq \frac{\eps\tau}{k}$, otherwise $\tilde{v}_{ij}^{(\tau)}  = 0$.\\}
		Set $\eps' \lar \frac{\eps^{p/2}}{k^{p/2 -1}}$.\\
		Run in parallel an instance of the $\ell_2$ heavy hitter algorithm (Algorithm \ref{alg:l2_hh_static}) with each site's frequency vector set as $\tilde{v}^{(i,\tau)}$ and accuracy parameter as $\eps'$. Denote this instance as $\alg_\tau$.
	}
	Concurrently, the coordinator collects each site's local $F_p$ value and calculate $\ell_p'(v)$. Choose $\tau$ such that $\tau \leq \ell_p'(v) \leq 2\tau$. Use the estimation of $\alg_\tau$ for output.
\end{algorithm}

\begin{proof}[Proof of Theorem~\ref{thm:l2_hh_tracking}]
	We first analyze the error. 
	
	\paragraph{Error analysis.} We focus on analyzing the error at a fixed time $t$. For each element $j$, the coordinator estimates $v_{ij}$ independently for each site $i$. Let $\hat{v}_{ij}$ be the corresponding estimate. We will show that it is unbiased and its variance is bounded by $O(\eps^2F^{(i)}(t))$, where $F^{(i)}(t)$ is the local $F_2$ value on site $i$. Then the variance of $\hat{v_j} := \sum_{i=1}^k \hat{v}_{ij}$ is $O(\eps^2 F_2(t))$ as desired. We focus on the current round of site $i$, in which the local $F_2$ increases from $F$ to $2F$. We use $w_{ij}$ to denote the number of $j$'s arrived during the current round. Assume we are in phase $s$ and interval $c$, and the increment of $v_{ij}$ during the current interval is $\Delta$. Let $v$ denote the value of $w_{ij}$ at the beginning of phase $s$, which is known by the coordinator (due to line~\ref{line:exact} in Algorithm~\ref{alg:l2_hh_tracking_site}). So the error only comes from the current phase. We can decompose $ w_{ij} - v = \sum_{h=1}^{c-1} 2^{h} + \Delta$, where $\Delta \le 2^c$. In each interval $h\in [c-1]$, site $i$ sends a message $(0,i,j,\frac{\eps^2 F}{2^h})$ with probability $\Pr[2^h \ge r_{s,h}] = \frac{2^h\cdot 2^h}{\eps^2 F}$. Define 
	$$\hat{w}_h = \mathbf{1} (\textrm{site $i$ sends a message in the $h$th interval})\cdot \frac{\eps^2 F}{2^{h}}.$$
	Then $\E[\hat{w}_h] = 2^h$ and $\var[\hat{w}_h] \le \eps^2 F$. Similarly, for the $c$th interval, we have 
	\begin{align*}
		\E[\hat{w}_c] = \frac{2^c\cdot \Delta}{\eps^2 F} \cdot \frac{\eps^2 F}{2^c} = \Delta,
	\end{align*}
	and
	\begin{align*}
		\var[\hat{w}_c] \le \left(\frac{\eps^2 F}{2^c} \right)^2 \cdot \frac{2^c\cdot \Delta}{\eps^2 F} = \frac{\eps^2 F\cdot \Delta}{2^c} \le \eps^2 F.
	\end{align*}
	In Algorithm \ref{alg:l2_hh_tracking_coordinator}, we essentially use $\sum_{h=1}^c \hat{w}_c$ as an estimate of $w_{ij}-v$. By the above analysis, we have $\E[\sum_{h=1}^c \hat{w}_c] = w_{ij}-v$ and  $\var[\sum_{h=1}^c \hat{w}_c] \le c\cdot \eps^2F \le \log(\eps \sqrt{F})\eps^2 F$. This is the variance of the estimate of $w_{ij}$ in the current round. Since $F$ is geometrically increasing over rounds, the variance from all previous rounds is also dominated by $\log(\eps \sqrt{F})\eps^2 F$. Adjusting $\eps' = \eps/\sqrt{2\log(\eps \sqrt{F})}$ in the beginning, we obtain an estimate of $v_{ij}$ with variance $\eps^2 F$.
	
	\paragraph{Communication cost.}
	We analyze the communication for one round on site $i$, in which local $F_2$ increases from $F$ to $2F$. For each $j$, the tracking period is divided into phases. During a complete phase, the value of $v_{ij}$ increases by $\eps \sqrt{F}$. Let $a, b$ be the start and end values of $v_{ij}$ of this phase. Then $b^2-a^2 = (b+a)(b-a)\ge (b-a)^2 = \eps^2 F$. As a result, each complete phase contributes $\eps^2F$ in the local $F_2$. Since the local $F_2$ on site $i$ doubles in a round, the total number of phases over all elements in a round is at most $1/\eps^2$. Note that there are at most two message sent by site $i$ in any complete phase; thus the communication cost for all complete phases on site $i$ is $\tilde{O}(1/\eps^2)$. Next we analyze the communication cost of unfinished phases. 
	
	For each $j$, its unfinished phase consists of at most $\log(\eps\sqrt{F})$ intervals. Let $c$ be the current interval of tracking $j$ in the current phase. This means $w_{ij}-v$, i.e.\ the increment of $v_{ij}$ in the current phase, is $\Theta(2^{c})$. From the sampling scheme in our algorithm, the probability that site $i$ sends a message during interval $h$ is $\frac{2^{2h}}{\eps^2F}$, and therefore, the expected number of messages sent by site $i$ during the unfinished phase of tracking $j$ is $O\left(\frac{2^{2c}}{\eps^2 F} \right) \le O\left(\frac{w_{ij}^2}{\eps^2 F}\right)$. Sum over all $j\in [n]$, the expected total communication in terms of messages is 
	\begin{align*}
		\sum_{j=1}^n O\left(\frac{w_{ij}^2}{\eps^2 F}\right) = O\left(\frac{1}{\eps^2}\right).
	\end{align*}
	Hence, the communication cost of site $i$ during a round is ${O}(\frac{1}{\eps^2}\log n)$ bits since each message can be encoded with $O(\log n)$bits. Recall that we need to rescale $\eps$ by a factor of $\sqrt{\log(\eps \sqrt{F})}=O(\sqrt{\log n})$, and there are at most $O(\log n)$ rounds and $k$ sites, the total communication over the entire tracking period is $O\left(\frac{\log^3 n}{\eps^2}\right)$ bits. 
\end{proof}

\begin{proof}[Proof of Theorem~\ref{thm:lp_hh_tracking}]
	Fix time $t \in [m]$, $j\in[n]$. Suppose $\tau \leq \hat{\ell_p'}(t) \leq 2\tau$. Since $\hat{\ell_p'}(t) \in (1\pm \frac{1}{2}) \ell_p'(t)$, it can be seen that $\frac{2}{3}\tau \leq \ell_p'(t) \leq 4\tau$. Let $\what{v^\tau_j}(t)$ be the estimation of $\alg_\tau$ at time $t$. Denote $\wtilde{\ell_2'^{(\tau)}}(t)$ as the $\ell_2'$ value of $\wtilde{v^\tau}$ at time $t$. By the definition of $\wtilde{v^\tau}$ and Theorem \ref{thm:l2_hh_tracking}, it is known that
	\begin{align*}
		|\what{v^\tau_j}(t) - v_j(t)| \leq  |\what{v^\tau_j}(t) - \wtilde{v^\tau_j}(t)| + \eps \tau  \leq \eps'\wtilde{\ell_2'^{\tau}}(t) + \frac{3}{2}\eps \ell_p'(t).
	\end{align*}
	Note that,
	\begin{align*}
		(\wtilde{\ell_2'^{\tau}}(t))^2 &\leq \sum_{v_{ij(t)}\geq\frac{\eps\tau}{k}} v_{ij}^2(t) \leq \frac{k^p}{\eps^p}\cdot\frac{F_p'(t)}{\tau^p}\cdot (\frac{\eps\tau}{k})^2 \\
        &= \frac{k^{p - 2}}{\eps^{p-2}} \cdot \left(\frac{\ell_p'(t)}{\tau}\right)^{p-2}\ell_p'^2(t) \leq 4^{p-2} \frac{k^{p - 2}}{\eps^{p-2}} \cdot \ell_p'^2(t).
	\end{align*}
	Hence $\eps'\wtilde{\ell_2'^{\tau}}(t) = \frac{\eps^{p/2}}{2^{p-2}(k^{p/2} - 1)} \wtilde{\ell_2'^{\tau}}(t) \leq \eps \cdot \ell_p'(t)$, which implies $|\what{v^\tau_j}(t) - v_j(t)| \leq 3\eps \ell_p'(t)$. Now consider the communication. The $\ell_p'$ tracking instance costs $O(k\log n)$ bits. By Theorem \ref{thm:l2_hh_tracking}, each of the $O(\log m)$ $\ell_2$ heavy hitter tracking instance costs $O(\frac{k}{(\eps')^2}\log^2 m\log n) = O(\frac{k^{p-1}}{\eps^p}\cdot \log^2 m\log n)$ bits. Thus the total communication cost is $O(\frac{k^{p-1}}{\eps^p}\cdot \log^3 m\log n)$.
\end{proof}

\begin{algorithm}
	\caption{$\ell_p$ heavy hitter tracking}
	\label{alg:lp_hh_tracking}
	\Parallel{
		Start an instance of $\ell_p'$ tracking with accuracy parameter as $1/2$, denote its estimation at $t$ as $\hat{\ell_p'}(t)$. \\
		\For{$\tau \in \{1, 2, 4, 8, \cdots, 2^{\lceil\log m\rceil}\}$}{
			\For{Site $i\in[k]$}{
				Define $\widetilde{v^{(\tau,i)}} = (\widetilde{v^{\tau}_{i1}}, \widetilde{v^{\tau}_{i2}},\ldots,\widetilde{v^{\tau}_{in}})$, where $\widetilde{v^{\tau}_{ij}} = v_{ij} -\frac{\eps\tau}{k}$ if $v_{ij} \geq \frac{\eps\tau}{k}$, otherwise $\widetilde{v^{\tau}_{ij}}  = 0$.\\
			}
			Start an instance of $\ell_2$ heavy hitter tracking algorithm on $(\widetilde{v^{(\tau, 1)}}, \widetilde{v^{(\tau, 2)}}, \ldots, \widetilde{v^{(\tau, k)}})$, set the accuracy parameter as $\eps' = \frac{\eps^{p/2}}{2^{p-2}(k^{p/2} - 1)}$. Denote this instance as $\alg_{\tau}$. 
		}
	}
	For time $t$, find $\tau$ such that $\tau \leq \hat{\ell_p'}(t) \leq 2\tau$, use the estimation given by $\alg_\tau$ at this time for output.
\end{algorithm}

\end{document}